\documentclass[a4paper,12pt,final]{article}
\usepackage[sfdefault]{biolinum}  
\usepackage{amsfonts,amsmath,amsthm,amssymb,bbm}
\usepackage{geometry}
\usepackage{setspace}
\usepackage{soul}
\usepackage[multiple]{footmisc}
\usepackage{setspace}
\usepackage{natbib}
\usepackage{tikz}
\usepackage[colorlinks,linkcolor=links,citecolor=cites,urlcolor=MyDarkBlue]{hyperref}
\usepackage{bbm}
\usepackage[justification=centering]{caption}
\usepackage[section]{placeins}
\usepackage{parskip}
\usepackage{nicefrac}
\usepackage{subcaption}
\usepackage[normalem]{ulem}
\usepackage{ragged2e}
\definecolor{MyDarkBlue}{rgb}{0,0.08,0.45}
\definecolor{cites}{HTML}{324b13}
\definecolor{links}{HTML}{1a663b}
\definecolor{MyLightMagenta}{cmyk}{0.1,0.8,0,0.1}
\definecolor{sblue}{HTML}{0049A9}
\definecolor{scyan}{HTML}{CBEAFC}
\definecolor{sred}{HTML}{B5595C}
\definecolor{sgreen}{HTML}{609B57}
\definecolor{spink}{HTML}{FFB0FF}

\newtheorem{theorem}{Theorem}
\newtheorem{lemma}{Lemma} 
\newtheorem{proposition}{Proposition}

\newtheorem{definition}{Definition} 
\newtheorem{corollary}{Corollary}

\newtheorem{remark}{Remark}
\newtheorem{observation}{Observation}

\newcommand{\aaction}{\ensuremath{a}}
\newcommand{\aactionb}{\ensuremath{\aaction^\prime}}
\newcommand{\aactionbb}{\ensuremath{\aaction^{\prime\prime}}}

\newcommand{\Actions}{\ensuremath{A}}
\newcommand{\Actionsb}{\ensuremath{B}}
\newcommand{\Actionsc}{\ensuremath{C}}
\newcommand{\payoff}{\ensuremath{u}}
\newcommand{\type}{\ensuremath{\theta}}
\newcommand{\typeb}{\ensuremath{\type^\prime}}
\newcommand{\Types}{\ensuremath{\Theta}}
\newcommand{\Posteriors}{\ensuremath{\Delta(\Types)}}
\newcommand{\belief}{\ensuremath{\mu}}

\newcommand{\prior}{\ensuremath{\belief_0}}
\newcommand{\marginalbase}{\ensuremath{\nu}}
\newcommand{\marginal}{\ensuremath{\marginalbase_0}}
\newcommand{\joint}{\ensuremath{\pi}}

\newcommand{\reals}{\ensuremath{\mathbb{R}}}

\newcommand{\bsplit}{\ensuremath{\tau}}
\newcommand{\base}{\ensuremath{G}}

\newcommand{\setplayers}{\ensuremath{[N]}}
\newcommand{\demandbase}{\ensuremath{d}}
\newcommand{\demand}{\ensuremath{\demandbase_{(\marginal,\bsplit)}}}
\newcommand{\graphp}{\ensuremath{G_P}}
\newcommand{\strat}{\ensuremath{\alpha}}

\newcommand{\simplex}{\ensuremath{\Delta}}
\newcommand{\opt}{\ensuremath{\simplex^*}}
\newcommand{\opta}{\ensuremath{\opt(\aaction)}}
\newcommand{\optab}{\ensuremath{\opt(\aactionb)}}

\newcommand{\cardstates}{\ensuremath{|\Types|}}

\newcommand{\normal}{\ensuremath{N}}
\newcommand{\direction}{\ensuremath{c}}
\newcommand{\Directions}{\ensuremath{C}}
\newcommand{\binding}{\ensuremath{B}}
\newcommand{\vectors}{\ensuremath{x}}
\newcommand{\vectorb}{\ensuremath{\vectors^\prime}}

\newcommand{\deviation}{\ensuremath{d}}

\newcommand{\deviationbb}{\ensuremath{\deviation_{\aactionb,\aactionbb}}}

\newcommand{\marginalp}{\ensuremath{\overline{\marginalbase}_0}}
\newcommand{\rnpayoff}{\ensuremath{\tilde{\payoff}}}
\newcommand{\Joints}{\ensuremath{\Pi}}
\newcommand{\bce}{\ensuremath{\mathrm{BCE}}}
\newcommand{\bceprior}{\ensuremath{\bce\left(\prior\right)}}
\newcommand{\scr}{\ensuremath{\sigma}}
\newcommand{\nplayers}{\ensuremath{N}}
\newcommand{\playerindex}{\ensuremath{i}}
\newcommand{\pair}{\ensuremath{(\prior,\marginal)}}

\newcommand{\Actionsi}{\ensuremath{\Actions_\playerindex}}

\newcommand{\payoffi}{\ensuremath{\payoff_\playerindex}}

\newcommand{\actioni}{\ensuremath{\aaction_\playerindex}}
\newcommand{\actionbi}{\ensuremath{\aactionb_\playerindex}}

\newcommand{\actionmi}{\ensuremath{\aaction_{-\playerindex}}}

\title{{\Huge{The Core of Bayesian Persuasion}\thanks{We are grateful to Yaron Azrieli, Denniz Kattwinkel, Emir Kamenica, Shengwu Li, Elliot Lipnowski, Meg Meyer, Stephen Morris, Jacopo Perego, Andrea Prat, John Rehbeck, and Vasiliki Skreta for thought-provoking questions and insightful discussions. We owe special thanks to Quitz\'e Valenzuela-Stookey, who suggested ideas that facilitated the proof of \autoref{theorem:bce-c}.}}
}
\author{Laura Doval\thanks{Columbia Business School and CEPR. Email: \href{mailto:laura.doval@columbia.edu}{\texttt{laura.doval@columbia.edu}}} \and Ran Eilat\thanks{Department of Economics, Ben-Gurion University. E-mail: \href{mailto:eilatr@bgu.ac.il}{\texttt{eilatr@bgu.ac.il}}}}
\begin{document}
\maketitle
\begin{abstract}
An analyst observes the frequency with which an agent takes actions, but not the frequency with which she takes actions conditional on a payoff relevant state. In this setting, we ask when the analyst can rationalize the agent's choices as the outcome of the agent learning something about the state before taking action. Our characterization marries the obedience approach in information design \citep{bergemann2016bayes} and the belief approach in Bayesian persuasion \citep{kamenica2011bayesian} relying on a theorem by \cite{strassen1965existence} and Hall's marriage theorem. We apply our results to  ring-network games and to identify conditions under which a data set is consistent with a public information structure in first-order Bayesian persuasion games. 

\bigskip


\end{abstract}

\section{Introduction}
Given a primitive payoff structure, information design  provides a framework for rationalizing outcomes as the result of non-cooperative play \emph{without} having to specify the players' information structure. For this reason, the seminal work of \cite{bergemann2016bayes} has spurred renewed interest among empirical scholars wishing to obtain identification and estimation results under a weaker set of information assumptions (see, for instance, \citealp{syrgkanis2017inference,magnolfi2019estimation,koh2022stable}). 

However, the weaker set of assumptions on the information structure comes at the cost of increasing demands on the data set available to the analyst. Indeed, the usual assumption in the literature is that the analyst is given a \emph{joint} distribution over payoff relevant states and action profiles. For instance, the literatures on rational inattention and stochastic choice usually assume that the analyst observes an agent's choices conditional on the state of the world (e.g., \citealp{caplin2015revealed,aguiar2018does}).  Given this data set, Bayes correlated equilibrium provides an easy to test set of conditions that the joint distribution over states and action profiles must satisfy in order to be consistent with the outcome of non-cooperative play under some information structure.

Oftentimes, however, the analyst's data set is more limited. The analyst may observe the distribution over the payoff relevant states of the world and the distribution over action profiles, but not the distribution over action profiles conditional on the state of the world.\footnote{Whereas state-dependent stochastic choice data is useful to guide the design and interpretation of experiments, this data is oftentimes hard to come by outside the experimental setting. \cite{dardanoni2020inferring} provides an eloquent discussion of the \emph{data voracity} of stochastic choice.} We can then ask, given the primitive payoff structure, which marginal distributions can be rationalized as the outcome of non-cooperative play under some information structure. We refer to such marginals as \emph{BCE consistent} because they satisfy that a joint distribution over states and action profiles exists that is consistent with the marginals \emph{and} is a Bayes correlated equilibrium. Characterizing the set of BCE-consistent marginal distributions can only increase the practical applicability of Bayes correlated equilibrium.

The set of BCE-consistent marginal distributions is of interest for two other reasons. First, the analyst oftentimes is not just interested in the existence of an information structure that rationalizes the (marginal) distribution of play, but one that satisfies certain properties. For instance, the analyst may want to test whether the agents have private information. As we explain below, our characterization result provides us with a test for the existence of a public information structure that rationalizes the observed distribution of play. The second reason is related to reduced-form implementation in mechanism design \citep{matthews1984implementability,border1991implementation}. Whenever the information designer only cares about the agents' action profiles, but not the state of the world, the information designer's problem can be expressed as the choice out of the set of BCE-consistent marginals.

In this paper, we take the first step towards characterizing the set of BCE-consistent marginals by considering the single-agent case. \autoref{theorem:bce-c} provides a characterization of the set of BCE-consistent marginals building on a theorem in \cite{strassen1965existence}. Furthermore, marrying the obedience approach in information design with the belief approach in \cite{kamenica2011bayesian}, \autoref{theorem: single agent BCE} characterizes the Bayes plausible distributions over posteriors that implement a given marginal over actions. We provide two network-based proofs of \autoref{theorem: single agent BCE}.  Relying on recent extensions of Hall's marriage theorem in \cite{barseghyan2021heterogeneous} and \cite{azrieli2022marginal}, the first characterization uncovers a connection between BCE-consistency and  the \emph{core} of the game induced by loosely speaking, some (Bayes plausible) posterior distribution (see \autoref{remark:core} and \citealp{grabisch2016set}). The second proof relies on the demand problem of \cite{gale1957theorem}. We show that one can interpret BCE-consistency problem as a supply-demand problem in a persuasion economy, in which the marginal action distribution describes the demand and a Bayes plausible posterior distribution describes the supply.  We then rely on the results in \cite{gale1957theorem} to determine when the demand is feasible given the supply.

\autoref{sec:applications} illustrates how \autoref{theorem:bce-c} already allows us to study multi-agent games. \autoref{sec:public} applies \autoref{theorem:bce-c} to the first-order Bayesian persuasion setting of \cite{arieli2021feasible} to characterize the subset of BCE-consistent marginals that are consistent with a \emph{public} information structure. Instead, \autoref{sec:ring} applies \autoref{theorem:bce-c} to characterize BCE-consistent marginals in ring-network games as in \cite{kneeland2015identifying}.


\paragraph{Related literature} The two closest papers to ours are \cite{rehbeck2023revealed} and \cite{azrieli2022marginal}. \cite{rehbeck2023revealed} studies the same question as us, but when the analyst has access to a decision maker's unconditional stochastic choices, possibly out of different menus. For the case of a single menu, the characterization in \cite{rehbeck2023revealed} is different from that in \autoref{theorem:bce-c} and is stated in terms of the non-existence of a possibly mixed deviation. \cite{azrieli2022marginal} study a similar question to ours in the context of stochastic choice. In their setting, the analyst has access to a marginal distribution over a decision maker's choices and a marginal distribution over the menus out of which the decision maker made her choices. \cite{azrieli2022marginal} show that the marginal distributions are consistent if and only if the marginal over choices is in the core of the game induced by the marginal over menus. 

A literature in decision theory and experimental economics studies when choices can be rationalized via costly information acquisition and whether the choices can be used to identify the information acquisition costs (see, e.g., \cite{caplin2015revealed}, \cite{caplin2017rationally}, \cite{chambers2020costly}, \cite{dewan2020estimating}, \cite{denti2022posterior}). Like we do, many of these papers assume that the decision maker's utility is known.  More recently, assuming that the analyst has access to state-dependent stochastic choice data, \cite{caplin2023rationalizable} study when choices can be rationalized as if the agent has access to some information before choosing her actions. Whereas their analyst has access to a richer data set, they require consistency of the information structure across a family of decision problems.

\cite{arieli2021feasible} and \cite{morris2020no} characterize joint distributions over posterior beliefs that are consistent with some information structure.\footnote{Whereas \cite{arieli2021feasible} study the binary-state case, the characterization in \cite{morris2020no} requires no such assumption.} Both papers cast the problem as one of distributions with given marginals: they take as given a profile of marginal distributions over posterior beliefs with the same mean and characterize when a joint distribution with the given marginals exists that is consistent with information. 

Finally,  \cite{toikka2022bayesian} study reduced-form implementation in a Bayesian persuasion in which the sender and the receiver care only about the posterior mean of the states. They leverage the mean preserving spread property to write a linear programming problem for the sender that only depends on the marginal distribution over actions. Beyond the posterior mean setting, they do not provide a characterization of the set of implementable marginal action distributions.

\section{Model}\label{sec:model}
Anticipating our multi-agent results in \autoref{sec:applications}, our notation below presumes multiple agents. We then specialize it to the single-agent case in \autoref{sec:single-agent}:

\paragraph{Base game:} An incomplete information \emph{base game}, \base, is defined as follows. We are given a set of \nplayers\ players, \setplayers=$\{1,\dots,\nplayers\}$. Each player $\playerindex\in\setplayers$ chooses an action from the finite set $\Actionsi$. Payoffs $\payoffi(\aaction,\type)$ depend on the action profiles $\aaction\in\Actions\equiv\times_{\playerindex\in\setplayers}\Actionsi$ and the state of the world, \type, an element of the finite set \Types.\footnote{As we explain in \autoref{sec:single-agent} our single-agent characterization extends to the case in which \Types\ and \Actions\ are infinite (see \autoref{remark:strassen}). However, the set of finitely many states and actions allows us to provide a \emph{sharper} characterization.}  The players share a common prior $\prior\in\Delta(\Types)$ over the state of the world. That is, $\base=\langle\Types,(\Actionsi,\payoffi)_{\playerindex\in\setplayers},\prior\rangle$.

\paragraph{Bayes correlated equilibrium:} An \emph{outcome} is a joint distribution over action profiles and states of the world, $\joint\in\Delta(\Actions\times\Types)$. We are concerned with those outcomes that are consistent with non-cooperative play of the base game, where the solution concept is Bayes Nash equilibrium. The notion of Bayes correlated equilibrium in \cite{bergemann2016bayes} captures the set of outcomes that are consistent with (Bayes Nash) equilibrium of the base game under \emph{some} information structure:

\begin{definition}[Bayes correlated equilibrium]\label{definition:bce} An outcome distribution $\joint\in\Delta(\Actions\times\Types)$ is a \emph{Bayes correlated equilibrium} of base game $\base=\langle\Types,(\Actionsi,\payoffi)_{\playerindex\in\setplayers},\prior\rangle$, if for all agents $\playerindex\in\setplayers$, actions $\actioni,\actionbi\in\Actionsi$, the following holds
\begin{align}\label{eq:obedience}\tag{O}
\sum_{(\actionmi,\type)}\joint(\actioni,\actionmi,\type)\left[\payoffi(\actioni,\actionmi,\type)-\payoffi(\actionbi,\actionmi,\type)\right]\geq0,
\end{align}
and for all $\type\in\Types$
\begin{align}\label{eq:martingale}\tag{M$_\Types$}
    \sum_{\aaction\in\Actions}\joint(\aaction,\type)=\prior(\type).
\end{align}
Let \bceprior\ denote the set of Bayes correlated equilibria.
\end{definition}
In words, a Bayes correlated equilibrium is an outcome distribution that satisfies a series of \emph{obedience} constraints \eqref{eq:obedience} and a \emph{martingale} condition \eqref{eq:martingale}. The first ensures each player's best response condition under \emph{some} information structure, whereas the second ensures the existence of an information structure that is consistent with the players' prior information. Note that any Bayes correlated equilibrium $\joint\in\Delta(\Actions\times\Types)$ induces two marginal distributions, $(\joint_\Types,\joint_\Actions)\in\Posteriors\times\Delta(\Actions)$. The definition of Bayes correlated equilibrium implies that the primitive base game \base\ pins down $\joint_\Types$, but not necessarily $\joint_\Actions$.

\paragraph{Information Design with Given Marginals:} We take the point of view of an analyst who knows the base game, but not the information structure under which the base game is played. The analyst is also endowed with information about the actions taken by the players. The analyst's goal is to determine whether this information is consistent with non-cooperative play of the base game under some information structure.

We consider two kinds of information the analyst may have about the players' actions, which are equivalent in the single-agent setting. In the first case, the analyst is endowed with a distribution over action profiles, $\marginal\in\Delta(\Actions)$. In the second case, the analyst is endowed with a profile of action distributions, one for each player, that is, $\marginalp=(\marginalbase_{0,1},\dots,\marginalbase_{0,N})\in\times_{\playerindex\in\setplayers}\Delta(\Actionsi)$. 

In each of these cases, the analyst wants to ascertain whether a Bayes correlated equilibrium $\joint\in\bceprior$ exists such that $\joint_\Actions$ coincides with the analyst's information about the players' actions (i.e., $\joint_\Actions=\marginal$ or $\times_{\playerindex\in\setplayers}\joint_{\Actionsi}=\marginalp$). In this case, we say that the marginals $\pair$ are BCE-consistent or that the profile of marginal distributions $(\prior,\marginalp)$ are M-BCE-consistent. \autoref{definition:bce-consistent} records this for future reference:
\begin{definition}[BCE- and M-BCE-consistent marginals]\label{definition:bce-consistent} Say that $\pair$ are \emph{BCE-consistent} if a Bayes correlated equilibrium  $\joint\in\bceprior$ exists such that $\joint_\Actions=\marginal$. Similarly, we say that $(\prior,\marginalp)$ are \emph{M-BCE-consistent} if a Bayes correlated equilibrium $\joint\in\bceprior$ exists such that for all players $\playerindex\in\setplayers$, $\joint_{\Actionsi}=\marginalbase_{0,i}$.
\end{definition}
Note that if $\pair$ are BCE-consistent, then letting $\marginalbase_{0,\playerindex}$ denote the marginal of \marginal\ over $\Actionsi$, we have that $(\prior,\marginalbase_{0,1},\dots,\marginalbase_{0,\nplayers})$ are M-BCE-consistent.

\paragraph{Constrained Optimal Transport} We close this section by noting a connection with optimal transport. Given $\pair$, let $\Joints\pair$ denote the set of joint distributions $\joint\in\Delta(\Actions\times\Types)$ withmarginals $\pair$, i.e.,  $(\joint_\Types,\joint_\Actions)=\pair$. Note that $\Joints\pair$ is always nonempty, e.g., the joint distribution $\joint(\aaction,\type)=\marginal(\aaction)\prior(\type)$ satisfies the marginal constraints. Instead, the subset $\Joints_\mathrm{O}\pair$ of $\Joints\pair$ that satisfies the obedience constraints \eqref{eq:obedience} may be empty. Thus, the characterization of the set of BCE-consistent marginals $\pair$ is equivalent to the characterization of when the feasible set of a \emph{constrained} optimal transport problem--in this case $\Joints_\mathrm{O}\pair$--is nonempty.\footnote{In their study of credible Bayesian persuasion, \cite{lin2022credible} characterize the set of credible outcome distributions by noting a connection with optimal transport. In their case, to check whether a given message distribution $\lambda_M$ is implementable, it must be that no other joint distribution over states and messages that respects the given marginals exists and is preferred by the sender to $\lambda_M$.}

\section{Single-agent case}\label{sec:single-agent}
In this section we characterize the set of BCE-consistent marginals in the case of a single agent, that is, $\nplayers=1$. For this reason, in what follows we remove the index $i=1$ from the action set and the utility function. 

\paragraph{Distributions over posteriors and stochastic choice} An outcome distribution $\joint\in\Delta(\Actions\times\Types)$ with marginals $\pair$ induces two conditional probability systems: The first, $\{\belief(\cdot|\aaction)\in\Posteriors:\aaction\in\Actions\}$, describes the agent's beliefs conditional on action \aaction\ and satisfies for all actions $\aaction\in\Actions$,
\[\marginal(\aaction)\belief(\type|\aaction)=\joint(\aaction,\type).\]
In this case, one can view \marginal\ as a distribution over posteriors and the belief system $\left(\belief(\cdot|\aaction)\right)_{\aaction\in\Actions}$ as its support. 

The second, $\{\scr(\cdot|\type)\in\Delta(\Actions):\type\in\Types\}$, describes the agent's actions conditional on state \type\ and satisfies for all states $\type\in\Types$,
\[\prior(\type)\scr(\aaction|\type)=\joint(\aaction,\type).\]
The collection $\{\scr(\cdot|\type):\type\in\Types\}$ is what the stochastic choice literature dubs the agent's stochastic choice rule.

The analysis that follows characterizes the set of BCE-consistent marginals relying on the belief system, $\{\belief(\cdot|\aaction)\in\Posteriors:\aaction\in\Actions\}$. Instead, the stochastic choice rule $\scr(\cdot|\type)$ is the focus of the analysis in \autoref{sec:core}.

\paragraph{The action marginal as a distribution over posteriors} Given marginals $\pair$, the goal is to determine whether a belief system $\{\belief(\cdot|\aaction):\aaction\in\Actions\}$ exists that satisfies for all states $\type\in\Types$
\begin{align}\label{eq:marginal-states-b}\tag{BP$_{\prior}$}
    \sum_{\aaction\in\Actions}\marginal(\aaction)\belief(\type|\aaction)=\prior(\type),
\end{align}
and for all $\aaction,\aactionb\in\Actions$,
\begin{align}\label{eq:obedience-b}\tag{O$_\belief$}
\sum_{\type\in\Types}\marginal(\aaction)\belief(\type|\aaction)\left[\payoff(\aaction,\type)-\payoff(\aactionb,\type)\right]\geq0.
\end{align}
For an action \aaction, let $\opta$ denote the set of beliefs under which \aaction\ is optimal.\footnote{Formally, $\opta=\{\belief\in\Posteriors:(\forall\aactionb\in\Actions)\sum_{\type\in\Types}\belief(\type)\left(\payoff(\aaction,\type)-\payoff(\aactionb,\type)\right)\geq0\}$.} 
Then, Equations \ref{eq:marginal-states-b} and \ref{eq:obedience-b} require that (i) \marginal\ induces a Bayes plausible distribution over posteriors and (ii) for all actions \aaction, the \emph{posterior belief} $\belief(\cdot|\aaction)$ is an element of $\opta$. Under this interpretation,  the action distribution \marginal\ describes the frequency with which inducing beliefs in $\opta$ is necessary. Unsurprisingly, some of the conditions in \autoref{theorem:bce-c} below also check that \marginal\ satisfies a version of the martingale condition \citep{aumann1995repeated,kamenica2011bayesian}.

\autoref{theorem:bce-c} characterizes the set of BCE-consistent marginals:

\begin{theorem}[BCE-consistency]\label{theorem:bce-c}
The pair $\pair$ is BCE-consistent if and only if for all states $\type\in\Types$,
\begin{align}
\sum_{\aaction\in\Actions}\marginal(\aaction)\min_{\belief\in\opta}\belief(\type)&\leq\prior(\type),\label{eq:bce-c-states}
\intertext{and for all pairs of actions $\aactionb,\aactionbb\in\Actions$,}
\label{eq:bce-c-actions}
\sum_{\aaction\in\Actions}\marginal(\aaction)\max_{\belief\in\opta}\sum_{\type\in\Types}\belief(\type)\left[u(\aactionb,\type)-u(\aactionbb,\type)\right]&\geq\sum_{\type\in\Types}\prior(\type)\left[u(\aactionb,\type)-u(\aactionbb,\type)\right].
\end{align}
\end{theorem}
The proof is in \autoref{appendix:omitted}. In what follows, we provide intuition for the statement in \autoref{theorem:bce-c} and review the main steps of its proof.

\autoref{eq:bce-c-states} can be interpreted through the lens of the martingale property of beliefs. As discussed before \autoref{theorem:bce-c}, the action distribution \marginal\ describes the frequency with which beliefs in $\opta$ must be induced to satisfy \eqref{eq:marginal-states-b}. For a given state $\type\in\Types$, the term 
\[\underline{\belief}_\aaction(\type)\equiv\min_{\belief\in\opta}\belief(\type),\]
describes the smallest probability that the agent can assign to state \type\ and action \aaction\ be optimal. Thus,  \autoref{eq:bce-c-states} states that for $\pair$ to be BCE-consistent, it must be that the average under \marginal\ of these minimum probabilities, $\underline{\belief}_\aaction(\type)$, are below the prior probability of $\type$, \prior(\type). It is immediate that if for some state \type, \autoref{eq:bce-c-states} does not hold, then $\pair$ cannot be BCE-consistent. 

As we argue next, \autoref{eq:bce-c-actions} can be interpreted through the lens of a martingale property for the utility differences, $\payoff(\aactionb,\type)-\payoff(\aactionbb,\type)$. That is, for all pairs of actions, $\aactionb,\aactionbb$, the agent's expected ranking over \aactionb\ and \aactionbb\ under the experiment that rationalizes $\pair$ has to coincide with the agent's ex ante ranking over these actions, which is the right-hand side of \autoref{eq:bce-c-actions}. Indeed, because \autoref{eq:bce-c-actions} must hold when we exchange the roles of \aactionb\ and \aactionbb, we obtain that $\pair$ must also satisfy that
\begin{align}\label{eq:bce-c-actions-2}\sum_{\aaction\in\Actions}\marginal(\aaction)\min_{\belief\in\opta}\sum_{\type\in\Types}\belief(\type)\left[u(\aactionb,\type)-u(\aactionbb,\type)\right]\leq\sum_{\type\in\Types}\prior(\type)\left[u(\aactionb,\type)-u(\aactionbb,\type)\right].\end{align}
That is, the ranking at the prior between \aactionb\ and \aactionbb\ must be in between the worst and best rankings under the ``distribution over posteriors'' \marginal.

This is most easily seen in the simple case that \aactionbb\ is strictly optimal at the prior and $\{\aactionb,\aactionbb\}$ are the only actions in the support of \marginal.  Because \aactionb\ is in the support of \marginal, under a BCE \joint\ that satisfies the marginal constraints the agent must sometimes find it optimal to take action \aactionb\ instead of \aactionbb. Note, however, that \emph{on average} it must be the case that the agent finds action \aactionbb\ better than \aactionb. Consequently, under \joint, when the agent takes \aactionbb, the agent must prefer \aactionbb\ over \aactionb\ (weakly) more than at the prior. Because the left-hand side of \autoref{eq:bce-c-actions} selects beliefs in favor of \aactionb, it is immediate that if \autoref{eq:bce-c-actions} fails one cannot find an experiment in which the agent would take action \aactionb\ with sufficiently high probability so as to match \marginal.

So far, we have argued that the conditions in \autoref{theorem:bce-c} are necessary for $\pair$ to be BCE-consistent. To explain why they are also sufficient, it is useful to review the main steps in the proof of \autoref{theorem:bce-c}. Key to our proof is the following result from \cite{strassen1965existence}, which we record in present notation:

\begin{observation}[\protect{\citet[Theorem 3 and Corollary 1]{strassen1965existence}}]\label{theorem:strassen} 
A conditional probability system $\{\belief(\cdot|\aaction)\in\Posteriors:\aaction\in\Actions\}$ exists such that
\begin{enumerate}
    \item For all actions $\aaction\in\Actions$, $\belief(\cdot|\aaction)\in\opta$, and 
    \item For all states $\type\in\Types$, \ref{eq:marginal-states-b} holds,
\end{enumerate}
if and only if for all directions $\direction\in\reals^{\cardstates}$,
\begin{align}\label{eq:strassen}
    \sum_{\aaction\in\Actions}\marginal(\aaction)\max\{\direction^T\belief:\belief\in\opta\}\geq \direction^T\prior.
\end{align}
\end{observation}
Whereas Theorem 3 in \cite{strassen1965existence} requires that \autoref{eq:strassen} holds for \emph{all} directions in $\reals^{\cardstates}$, \autoref{theorem:bce-c} states that verifying \autoref{eq:strassen} holds for \emph{finitely} many directions is enough to conclude that $\pair$ are BCE-consistent. To see this, note that Equations \ref{eq:bce-c-states} and \ref{eq:bce-c-actions} correspond to \autoref{eq:strassen} for specific directions $\direction\in\reals^{\cardstates}$. Indeed, \autoref{eq:bce-c-states} corresponds to  $\direction=-e_\type\in\reals^{\cardstates}$, where $e_\type$ is the vector with a $1$ in the \type-coordinate and $0$ otherwise. Instead, \autoref{eq:bce-c-actions} corresponds to the direction $\direction=-\deviationbb$, where \deviationbb\ is the vector with \type-coordinate $\deviationbb (\type)=\payoff(\aactionb,\type)-\payoff(\aactionbb,\type)$.  

To see why verifying that \autoref{eq:strassen} holds for directions $\{(-e_\type)_{\type\in\Types},(-\deviationbb)_{\aactionb,\aactionbb\in\Actions}\}$ is enough to determine that \autoref{eq:strassen} holds for all directions $\direction\in\reals^{\cardstates}$, note the following. First, for a fixed action \aactionb, the directions $\{(-e_\type)_{\type\in\Types},(-\deviationbb)_{\aactionbb\in\Actions}\}$ are the normal vectors that define the polyhedron \optab. Indeed, the directions $(-e_\type)_{\type\in\Types}$ correspond to the condition that the elements of \optab\ are non-negative, whereas the directions $(-\deviationbb)_{\aactionbb\in\Actions}$ correspond to the condition that action \aactionb\ is optimal for all beliefs in \optab. Second, it is immediate that in each of the maximization problems on the left hand side of \autoref{eq:strassen}, the maximum is attained at an extreme point of \opta. Standard results in convex analysis then imply that if \autoref{eq:strassen} holds at all normal directions defining the polyhedra $\{\opta:\aaction\in\Actions\}$, then it holds for all directions (cf. \citealp{hiriart2004fundamentals}).

We close \autoref{sec:single-agent} with a remark on the generality of the results in \cite{strassen1965existence}. It can be skipped with no loss of continuity. 

\begin{remark}[\citealp{strassen1965existence}]\label{remark:strassen}
Theorem 3 and Corollary 1 in \cite{strassen1965existence} hold more generally than our current assumptions. In present notation, Corollary 1 applies whenever (i) $\Types$ and \Actions\ are compact metric spaces and the mapping $\aaction\mapsto\Delta^*(\aaction)$ from \Actions\ to subsets of \Posteriors\ is such that $\cup_{\aaction\in\Actions}\{\aaction\}\times\Delta^*(\aaction)$ is closed within $\Actions\times\Posteriors$ endowed with the weak$^*$-topology.\footnote{Instead, \citet[Theorem 3]{strassen1965existence} requires that \Types\ is Polish, \Actions\ be a convex compact topological vector space, and an appropriate measurability condition on the mapping  $\aaction\mapsto\sup\{\int \direction(\type)\belief(d\type):\belief\in\opta\}$ for any continuous function $\direction$ on \Types.}

In other words, under the aforementioned assumptions, (an integral version of) \autoref{eq:strassen} characterizes the set of BCE-consistent marginals.\footnote{To be precise, \autoref{eq:strassen} now becomes for all continuous functions $c:\Types\mapsto\reals$,
\[\int_{\Types}c(\type)\prior(d\type)\leq\int_\Actions\sup\left\{\int_\Types c(\type)\belief(d\type):\belief\in\opta\right\}\marginal(d\aaction)\].} The finite model allows us to provide a sharper characterization by reducing the number of directions one needs to consider.
\end{remark}

\subsection{The core of Bayesian Persuasion}\label{sec:core}
 
 In this section we provide a different perspective on \autoref{theorem:bce-c}. Together with the marginal distributions, $\pair$, we are given a distribution over posteriors $\bsplit\in\Delta(\Posteriors)$ with mean equal to the prior \prior. \autoref{theorem: single agent BCE} below characterizes the set of such distributions over posteriors that can \emph{implement} the marginal \marginal. Whereas this characterization does not substitute that in \autoref{theorem:bce-c}, it allows us to illustrate how one would go about constructing an information structure that implements \marginal. Along the way we also establish formal connections with the literature on stochastic choice. For this reason, we work with the agent's stochastic choice rule $\{\scr(\cdot|\type):\type\in\Types\}$ instead of the belief system $\{\belief(\cdot|\aaction):\aaction\in\Actions\}$. 
 
\paragraph{Obedient stochastic choice} To understand the results that follow, it is useful to state the obedience and marginal conditions in terms of the stochastic choice rule: Given $\pair$, we want a stochastic choice rule that satisfies for all actions $\aaction\in\Actions$
\begin{align}\label{eq:marginal-actions-b}\tag{M$_{\Actions}$}
    \sum_{\type\in\Types}\prior(\type)\scr(\aaction|\type)=\marginal(\aaction),
\end{align}
and for all $\aaction,\aactionb\in\Actions$,
\begin{align}\label{eq:obedience-c}\tag{O$_{\scr}$}
\sum_{\type\in\Types}\prior(\type)\scr(\aaction|\type)\left[\payoff(\aaction,\type)-\payoff(\aactionb,\type)\right]\geq0.
\end{align}

 \paragraph{Distributions over posteriors and stochastic choice} Given a Bayes plausible distribution over posteriors $\bsplit\in\Delta(\Posteriors)$, constructing a state-dependent stochastic choice rule is almost at hand. Almost because a Bayes plausible distribution over posteriors does not specify how the agent breaks ties when indifferent. Indeed, to a Bayes plausible distribution over posteriors, $\bsplit(\belief)$, we can associate a decision rule $\strat:\Posteriors\mapsto\Delta(\Actions)$, describing the probability $\strat(\aaction|\belief)$ with which the agent takes action \aaction\ when her belief is \belief. The pair $(\bsplit,\strat)$ determines a stochastic choice rule $\{\scr(\cdot|\type):\type\in\Types\}$ as follows:
\begin{align}\label{eq:tau-feasibility}
    \scr(\aaction|\type)=\sum_{\belief\in\Posteriors}\bsplit(\belief)\frac{\belief(\type)}{\prior(\type)}\strat(\aaction|\belief).
\end{align}
\autoref{eq:tau-feasibility} suggests that conditions under which a stochastic choice rule that satisfies \ref{eq:marginal-actions-b} and \ref{eq:obedience-c} are intimately related to the existence of a Bayes plausible distribution \bsplit\ and a decision rule \strat\ that satisfy certain properties. In fact, the analysis that follows identifies conditions on Bayes plausible distributions over posteriors under which a decision rule exists that induces a stochastic choice rule--and hence a joint distribution $\joint\in\Delta(\Actions\times\Types)$--that satisfies all the constraints.

\paragraph{Distributions over posteriors as distributions over menus} 
Given a Bayes plausible  $\bsplit\in\Delta(\Posteriors)$, one can construct a measure over subsets \Actionsb\ of the set of actions \Actions\ as follows. For each $\belief\in\Posteriors$, let $\aaction^*(\belief)$ denote the agent's best response when her belief is \belief. That is, $\aaction^*(\belief)=\arg\max_{\aaction\in\Actions}\mathbb{E}_{\type\sim\belief}\left[\payoff(\aaction,\type)\right]$. For each $\Actionsb\subseteq\Actions$, define $\bsplit_\Actions(\Actionsb)$ as 
\begin{align}\label{eq:tau-consideration-sets}
    \bsplit_\Actions(\Actionsb)=\bsplit\{\belief\in\Posteriors:\aaction^*(\belief)=\Actionsb\}.
\end{align}
In words, each action subset \Actionsb\ has mass equal to the probability that \bsplit\ induces a belief under which \Actionsb\ is optimal.

\autoref{theorem: single agent BCE} characterizes when the distribution over posteriors \bsplit\ implements \marginal:

\begin{proposition}\label{theorem: single agent BCE}
Suppose $\pair$ are BCE-consistent. A Bayes plausible distribution over posteriors, $\bsplit\in\Delta(\Posteriors)$, implements $\marginal$ if and only if for all $\Actionsb\subseteq\Actions$, the following holds
    \begin{align}\label{eq:core}
\sum_{\aaction\in\Actionsb}\marginal(\aaction)\geq\sum_{\Actionsc\subseteq\Actionsb}\bsplit_\Actions(\Actionsc).
    \end{align}
\end{proposition}
To interpret \autoref{eq:core}, note the following. The left-hand side of \autoref{eq:core} is the probability under which the agent takes \emph{some} action \aaction\ in the set \Actionsb. Instead, the right-hand side of \autoref{eq:core} is the probability under which the agent finds optimal \emph{some} action in the set \Actionsb\ (but no action that is not in \Actionsb). \autoref{eq:core} then says that the frequency with which the agent takes actions in \Actionsb\ has to be at least the frequency with which an action in \Actionsb\ is optimal.

\begin{remark}[A core interpretation]\label{remark:core} \autoref{eq:core} implies that \marginal\ is in the \emph{core of the game} induced by the measure $\bsplit_\Actions$.\footnote{\cite{azrieli2022marginal} also note the connection between stochastic menu choice and cooperative games.} Indeed, given $\bsplit_\Actions$, define the \emph{cooperative game} $(\Actions,w_{\bsplit_\Actions})$ as follows. The set function $w_{\bsplit_\Actions}:2^{\Actions}\mapsto\mathbb{R}$ is given by $w_{\bsplit_{\Actions}}(\Actionsb)=\sum_{\Actionsc\subseteq\Actionsb}\bsplit_\Actions(\Actionsc)$. Because $w_{\bsplit_{\Actions}}\geq0$, the core of the game $(\Actions,w_{\bsplit_\Actions})$ is given by
 \[\mathrm{Core}(w_{\bsplit_{\Actions}})=\left\{p\in\Delta\left(\Actions\right):(\forall\Actionsb\subseteq\Actions)\sum_{\aaction\in\Actionsb}p(\aaction)\geq w_{\bsplit_\Actions}(\Actionsb)\right\}.\] 
 \end{remark}

The proof of \autoref{theorem: single agent BCE} is based on the following graphical representation of the BCE-consistency problem depicted in \autoref{fig:graph-theorem-1}. Consider the following graph. Nodes are (i) the actions $\aaction\in\Actions$, (ii) the (non-empty) action subsets $\Actionsb\subseteq\Actions$ (i.e., the elements of $2^\Actions\setminus\{\emptyset\}$), (iii) a source node $s$, and (iv) a sink node $t$. Edges are as follows. There is an edge of weight one between $\aaction\in\Actions$ and $\Actionsb\subseteq\Actions$ if and only if $\aaction\in\Actionsb$. There is an edge with weight $\marginal(\aaction)$ between the source $s$ and \aaction. Finally, there is an edge between $\Actionsb\subseteq\Actions$ and the sink $t$ with weight $\bsplit_\Actions(\Actionsb)$. The condition in \autoref{eq:core} ensures that a feasible flow exists throughout the network.\footnote{That is, a flow $f$ such that $\sum_{\aaction\in\Actions}f(s,\aaction)=\sum_{\Actionsb\in 2^\Actions}f(\Actionsb,t)=1$.}

\begin{figure}[t]
\centering
\resizebox{14cm}{!}{\input{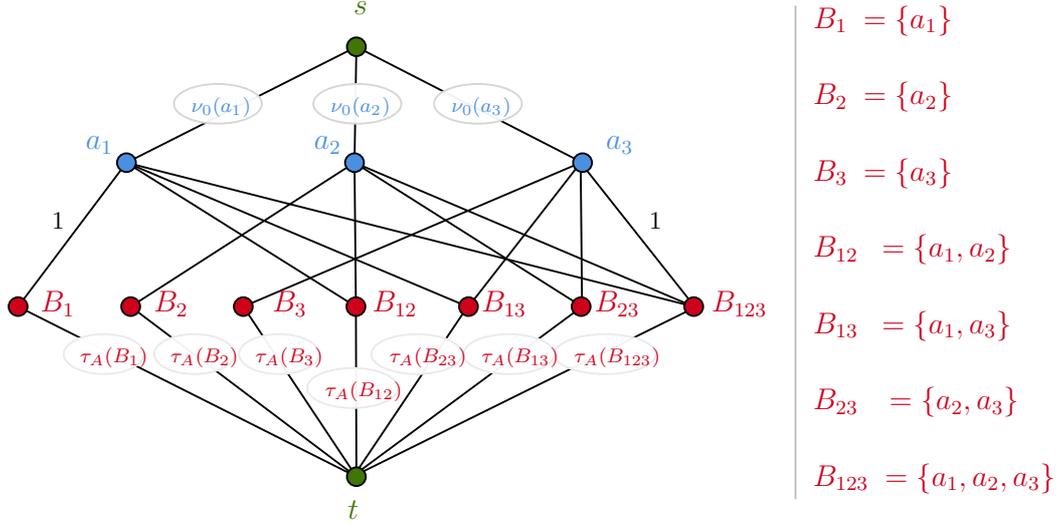}}
\caption{Graphical representation of the BCE-consistency problem with $3$ actions.}\label{fig:graph-theorem-1}
\end{figure}

\begin{proof}[Proof of \autoref{theorem: single agent BCE}]
It is immediate to show that if $\pair$ are BCE-consistent, then a Bayes plausible distribution over posteriors \bsplit\ exists such that \autoref{eq:core} holds.    

Suppose that \autoref{eq:core} holds for all $\Actionsb\subseteq\Actions$. 
    \citet[Proposition 9]{azrieli2022marginal} implies that a conditional probability system $\strat^\prime:2^\Actions\mapsto\Delta(\Actions)$ exists such that for all $\aaction\in\Actions$ 
    \begin{align}\label{eq:marriage-theorem}
       \marginal(\aaction)=\sum_{B:\aaction\in B}\bsplit_\Actions(\Actionsb)\strat^\prime(\aaction|\Actionsb).
    \end{align}
    The slight abuse of notation in the definition of the conditional probability system is justified since $\strat^\prime$ below plays the role of the decision rule in \autoref{eq:tau-feasibility}.
    
    We use the conditional probability system to create a s stochastic choice rule $\scr:\Types\mapsto\Delta(\Actions)$ as follows:
    \begin{align*}
\scr(\aaction|\type)=\sum_{B:\aaction\in B}\sum_{\belief:a^*(\belief)=B}\frac{\belief(\type)}{\prior(\type)}\bsplit(\belief)\strat^\prime(\aaction|B).
    \end{align*}
    The experiment has an intuitive explanation: We first draw a subset of actions \Actionsb\ using the measure $\bsplit_\Actions$ and then recommend to the agent which particular action she must take using the conditional probability system $\strat^\prime(\cdot|\Actionsb)$.

Define the information structure, $\joint\in\Delta(\Actions\times\Types)$ by letting $\joint(\aaction,\type)=\prior(\type)\scr(\aaction|\type)$.  To see that it has the desired properties, note first that
\begin{align*}
&\sum_{\aaction\in\Actions}\scr(\aaction|\type)=\sum_{\aaction\in\Actions}\sum_{\Actionsb:\aaction\in B}\sum_{\belief:a^*(\belief)=\Actionsb}\frac{\belief(\type)}{\prior(\type)}\bsplit(\belief)\strat^\prime(\aaction|\Actionsb)\\
&=\sum_{\Actionsb\subseteq\Actions}\left(\sum_{\aaction\in\Actionsb}\strat^\prime(\aaction|\Actionsb)\right)\sum_{\belief:\aaction^*(\belief)=\Actionsb}\bsplit(\belief)\frac{\belief(\type)}{\prior(\type)}
=\sum_{\Actionsb\subseteq \Actions}\sum_{\belief:\aaction^*(\belief)=\Actionsb}\bsplit(\belief)\frac{\belief(\type)}{\prior(\type)}=1
\end{align*}

Second, note that
\begin{align*}
\sum_{\type\in\Types}\joint(\aaction,\type)&=\sum_{\type\in\Types}\prior(\type)\scr(\aaction|\type)=\sum_{\type\in\Types}\sum_{\Actionsb:\aaction\in\Actionsb}\sum_{\belief:\aaction^*(\belief)=B}\belief(\type)\bsplit(\belief)\strat^\prime(\aaction|\Actionsb)\\
&=\sum_{\Actionsb:\aaction\in\Actionsb}\sum_{\belief:\aaction^*(\belief)=\Actionsb}\left(\sum_{\type\in\Types}\belief(\type)\right)\bsplit(\belief)\strat^\prime(\aaction|\Actionsb)=\marginal(\aaction),
\end{align*}
    by \autoref{eq:marriage-theorem}.

    Finally, note that the experiment is obedient: If \aaction\ is recommended with positive probability, then a set \Actionsb\ exists such that $\aaction\in\Actionsb$ and $\belief$ such that $\aaction^*(\belief)=\Actionsb$ is in the support of \bsplit, under which \aaction\ is optimal. Because $\scr(\aaction|\type)$ is obtained by averaging over beliefs in which \aaction\ is optimal, it remains optimal.
\end{proof}
\autoref{appendix:demand-supply} provides an alternative proof of \autoref{theorem: single agent BCE} using \citeauthor{gale1957theorem}'s network flow theorem.
\paragraph{Connection to stochastic choice:} The proof of \autoref{theorem: single agent BCE} connects two sets of conditional distribution over choices that arise in the stochastic choice literature: stochastic choices conditional on a state of the world--denoted by \scr\ in the proof--and stochastic choices out of a menu--denoted by $\strat^\prime$ in the proof. Indeed, the measure $\bsplit_\Actions$ can be interpreted as the frequency with which the agent faces different menus--action subsets in this case--whereas the measure \marginal\ represents the frequency with which the agent makes different choices. In other words, the pair $(\bsplit_\Actions,\marginal)$ is analogous to the data set in \cite{azrieli2022marginal}. Our ultimately goal, however, is to obtain the agent's stochastic choice rule, which we obtain relying on the Bayes' plausibility  of $\bsplit$.

\section{Applications}\label{sec:applications}
We consider in this section two applications of \autoref{theorem:bce-c} to simple multi-agent settings. \autoref{sec:public} studies under what conditions a pair of marginal distributions $\pair$ can be rationalized by a \emph{public} information structure. \autoref{sec:ring} shows that \autoref{theorem:bce-c} characterizes the set of M-BCE-consistent marginals.

\subsection{When is information public?}\label{sec:public}
We consider in this section the following multiplayer game. We assume $\nplayers\geq 1$ and that each player's utility function depends only her own action and the state of the world.\footnote{\cite{arieli2021feasible} dub this setting first-order Bayesian persuasion.} That is, for all players $\playerindex\in\{1,\dots,N\}$, all action profiles $(\actioni,\actionmi)\in\Actions$, and states of the world $\type\in\Types$,
\[\payoffi\left( \actioni,\actionmi, \type \right) =\payoffi\left( \actioni, \type \right). \] 

The analyst, who knows the base game \base\ and the marginal distribution of play $\marginal\in\Delta(\Actions)$,  wants to ascertain whether the distribution of play \marginal\ can be rationalized by a \emph{public} information structure (i.e., the players publicly observe the realization of a common signal structure before play).

As we show next, \autoref{theorem:bce-c} can be applied to address this question. In what follows, we rely on the following definition:

\begin{definition}[Public BCE-consistency]
The pair $\pair$ is \emph{public BCE-consistent} if: (i) $\pair$ are BCE-consistent, and (ii) a BCE $\joint\in\bceprior\cap\Joints\pair$ exists, whose information structure uses public signals alone.
\end{definition}

Consider now an auxiliary single-agent base game $\bar\base=\langle\Types, (\Actions,\bar u),\prior\rangle$. In this game, a player with payoff $\bar u(\aaction)=\sum_{\playerindex=1}^\nplayers\payoffi(\actioni,\type)$ chooses an action $\aaction\in\Actions=\times_{\playerindex\in\nplayers}\Actionsi$ under incomplete information about \type. 

The following result is an immediate corollary of \autoref{theorem:bce-c} and the focus on public signals:


\begin{corollary}\label{corollary: multi agent}
	$\pair$ are public BCE-consistent if and only if $\pair$ are BCE-consistent in base game $\bar\base$.
\end{corollary}
Because of the focus on public signal structures, the analysis of the multi-agent game reduces to the analysis of a single-agent problem. To see this, in a slight abuse of notation, let $\Actions^*(\belief)$ denote the set of actions that the agent with payoff $\bar\payoff$ finds optimal when their belief is \belief. It is immediate that $\Actions^*(\belief)=\times_{\playerindex\in\nplayers}\aaction_i^*(\belief)$, where for each player \playerindex, $\aaction_i^*(\belief)$ denotes the set of actions player \playerindex\ finds optimal when her belief is \belief. That is, the profile $\aaction=\left( \aaction_{1},\dots ,\aaction_\nplayers\right) \in \Actions$ is optimal for the agent with payoff $\bar\payoff$ if and only if action $\actioni$ is optimal for agent $i$, for all $\playerindex\in \setplayers$. And, given a posterior belief $\belief$, any distribution of (optimal) action profiles that the agent with payoff $\bar\payoff$ can generate, can also be generated by the players using a public correlation device or by duplicating signal realization, and vice versa.\footnote{For example, suppose that the signal realization $s$ induces the posterior belief $\belief$. Suppose also that under $\belief$, the agent with payoff $\bar\payoff$ selects the two optimal action profiles $\aaction,\aactionb\in \Actions^*(\belief)$ with equal probability. The same distribution of actions can be generated by the players: Indeed, one can \textquotedblleft split\textquotedblright\ the signal $s$ into two new signals, $s^{\prime }$ and $s^{\prime \prime }$, such that both new signals induce the same posterior belief $\belief$, and each of them is sent with half the probability of the original signal $s$. If whenever $s^{\prime }$ and $s^{\prime \prime }$ are realized, each agent acts according to her corresponding optimal action in the profiles $\aaction$
and $\aactionb$, respectively, the distribution over actions will coincide with that of the agent with payoff $\bar\payoff$.}
Notice that this equivalence no longer holds if either information is not public, or the players' utilities are interdependent.
\subsection{Ring-network games}\label{sec:ring}
We consider here ring-network games as in \cite{kneeland2015identifying}, extended to account for incomplete information. A ring-network game is a base game \base\ 
in which player's payoffs satisfy the following:
\begin{align}\label{eq:rn-game}\tag{RN-P}
    \payoff_1(\aaction,\type)&=\rnpayoff_1(\aaction_1,\type)\\
    (\forall \playerindex\geq2)\payoffi(\aaction,\type)&=\rnpayoff_\playerindex(\aaction_{\playerindex-1},\actioni).\nonumber
\end{align}
In words, player $1$ cares about their action and the state of the world, whereas for $i\geq2$ player $i$ cares about their action and that of player $i-1$. Ring-network games are used in the experimental literature that measures players' higher order beliefs to identify departures from Nash equilibrium. 

The analyst knows the ring-network base game and for each player $i$, player $i$'s action distribution, $\marginalbase_{0,\playerindex}\in\Delta(\Actionsi)$. The analyst wants to ascertain whether $(\prior,\marginalp)$ is M-BCE-consistent. Relying on \autoref{theorem:bce-c} and the ring-network structure, \autoref{proposition:ring-network} characterizes the set of M-BCE-consistent marginals: 

\begin{proposition}[M-BCE-consistency in ring-network games]\label{proposition:ring-network}
    The profile $(\prior,\marginalp)$ is M-BCE-consistent for the ring-network game $(\rnpayoff_\playerindex)_{\playerindex=1}^\nplayers$ if and only if the following holds:
    \begin{enumerate}
        \item $(\prior,\marginalbase_{0,1})$ are BCE-consistent in the base game $\langle\Types,\Actions_1,\rnpayoff_1,\prior\rangle$, 
        \item For all $i\geq2$, $(\marginalbase_{0,\playerindex-1},\marginalbase_{0,\playerindex})$ are BCE-consistent in the base game $\langle\Actions_{\playerindex-1},\Actionsi,\rnpayoff_\playerindex,\marginalbase_{0,\playerindex-1}\rangle$.
    \end{enumerate}
\end{proposition}
Similar to \autoref{corollary: multi agent}, \autoref{proposition:ring-network} exploits the structure of the ring-network game to reduce it to a series of single-agent problems in which except for player $1$, the states are given by the actions of the preceding player and the prior distribution over this state space by the marginal over actions of the preceding player.  Indeed, for $\playerindex\geq2$, BCE-consistency of $(\marginalbase_{0,\playerindex-1},\marginalbase_{0,\playerindex})$ implies that an information structure exists that rationalizes player $\playerindex$'s choices as the outcome of some information structure under ``prior'' $\marginalbase_{0,\playerindex-1}$, whereas BCE-consistency of $(\marginalbase_{0,\playerindex-2},\marginalbase_{0,\playerindex-1})$\footnote{With the understanding that $\marginalbase_{0,0}=\prior$.} guarantees that the ``prior'' $\marginalbase_{0,\playerindex-1}$ is consistent with player $\playerindex-1$ observing the outcome of some information structure given their belief $\marginalbase_{0,\playerindex-2}$.

\bibliographystyle{ecta}
\bibliography{id-core}

\begin{thebibliography}{29}
\newcommand{\enquote}[1]{``#1''}
\expandafter\ifx\csname natexlab\endcsname\relax\def\natexlab#1{#1}\fi

\bibitem[\protect\citeauthoryear{Aguiar, Boccardi, Kashaev, and Kim}{Aguiar
  et~al.}{2018}]{aguiar2018does}
\textsc{Aguiar, V.~H., M.~J. Boccardi, N.~Kashaev, and J.~Kim} (2018):
  \enquote{Does random consideration explain behavior when choice is hard?
  Evidence from a large-scale experiment,} \emph{arXiv preprint
  arXiv:1812.09619}.

\bibitem[\protect\citeauthoryear{Arieli, Babichenko, Sandomirskiy, and
  Tamuz}{Arieli et~al.}{2021}]{arieli2021feasible}
\textsc{Arieli, I., Y.~Babichenko, F.~Sandomirskiy, and O.~Tamuz} (2021):
  \enquote{Feasible joint posterior beliefs,} \emph{Journal of Political
  Economy}, 129, 2546--2594.

\bibitem[\protect\citeauthoryear{Aumann, Maschler, and Stearns}{Aumann
  et~al.}{1995}]{aumann1995repeated}
\textsc{Aumann, R.~J., M.~Maschler, and R.~E. Stearns} (1995): \emph{Repeated
  games with incomplete information}, MIT press.

\bibitem[\protect\citeauthoryear{Azrieli and Rehbeck}{Azrieli and
  Rehbeck}{2022}]{azrieli2022marginal}
\textsc{Azrieli, Y. and J.~Rehbeck} (2022): \enquote{Marginal stochastic
  choice,} \emph{arXiv preprint arXiv:2208.08492}.

\bibitem[\protect\citeauthoryear{Barseghyan, Coughlin, Molinari, and
  Teitelbaum}{Barseghyan et~al.}{2021}]{barseghyan2021heterogeneous}
\textsc{Barseghyan, L., M.~Coughlin, F.~Molinari, and J.~C. Teitelbaum} (2021):
  \enquote{Heterogeneous choice sets and preferences,} \emph{Econometrica}, 89,
  2015--2048.

\bibitem[\protect\citeauthoryear{Bergemann and Morris}{Bergemann and
  Morris}{2016}]{bergemann2016bayes}
\textsc{Bergemann, D. and S.~Morris} (2016): \enquote{Bayes correlated
  equilibrium and the comparison of information structures in games,}
  \emph{Theoretical Economics}, 11, 487--522.

\bibitem[\protect\citeauthoryear{Border}{Border}{1991}]{border1991implementation}
\textsc{Border, K.~C.} (1991): \enquote{Implementation of reduced form
  auctions: A geometric approach,} \emph{Econometrica: Journal of the
  Econometric Society}, 1175--1187.

\bibitem[\protect\citeauthoryear{Caplin and Dean}{Caplin and
  Dean}{2015}]{caplin2015revealed}
\textsc{Caplin, A. and M.~Dean} (2015): \enquote{Revealed preference, rational
  inattention, and costly information acquisition,} \emph{American Economic
  Review}, 105, 2183--2203.

\bibitem[\protect\citeauthoryear{Caplin, Dean, and Leahy}{Caplin
  et~al.}{2017}]{caplin2017rationally}
\textsc{Caplin, A., M.~Dean, and J.~Leahy} (2017): \enquote{Rationally
  Inattentive Behavior: Characterizing and Generalizing Shannon Entropy,} Tech.
  rep., National Bureau of Economic Research.

\bibitem[\protect\citeauthoryear{Caplin, Martin, and Marx}{Caplin
  et~al.}{2023}]{caplin2023rationalizable}
\textsc{Caplin, A., D.~J. Martin, and P.~Marx} (2023): \enquote{Rationalizable
  Learning,} Tech. rep., National Bureau of Economic Research.

\bibitem[\protect\citeauthoryear{Chambers, Liu, and Rehbeck}{Chambers
  et~al.}{2020}]{chambers2020costly}
\textsc{Chambers, C.~P., C.~Liu, and J.~Rehbeck} (2020): \enquote{Costly
  information acquisition,} \emph{Journal of Economic Theory}, 186, 104979.

\bibitem[\protect\citeauthoryear{Dardanoni, Manzini, Mariotti, and
  Tyson}{Dardanoni et~al.}{2020}]{dardanoni2020inferring}
\textsc{Dardanoni, V., P.~Manzini, M.~Mariotti, and C.~J. Tyson} (2020):
  \enquote{Inferring cognitive heterogeneity from aggregate choices,}
  \emph{Econometrica}, 88, 1269--1296.

\bibitem[\protect\citeauthoryear{Denti}{Denti}{2022}]{denti2022posterior}
\textsc{Denti, T.} (2022): \enquote{Posterior separable cost of information,}
  \emph{American Economic Review}, 112, 3215--3259.

\bibitem[\protect\citeauthoryear{Dewan and Neligh}{Dewan and
  Neligh}{2020}]{dewan2020estimating}
\textsc{Dewan, A. and N.~Neligh} (2020): \enquote{Estimating information cost
  functions in models of rational inattention,} \emph{Journal of Economic
  Theory}, 187, 105011.

\bibitem[\protect\citeauthoryear{Gale}{Gale}{1957}]{gale1957theorem}
\textsc{Gale, D.} (1957): \enquote{A theorem on flows in networks,}
  \emph{Pacific J. Math}, 7, 1073--1082.

\bibitem[\protect\citeauthoryear{Grabisch et~al.}{Grabisch
  et~al.}{2016}]{grabisch2016set}
\textsc{Grabisch, M. et~al.} (2016): \emph{Set functions, games and capacities
  in decision making}, vol.~46, Springer.

\bibitem[\protect\citeauthoryear{Hiriart-Urruty and
  Lemar{\'e}chal}{Hiriart-Urruty and
  Lemar{\'e}chal}{2004}]{hiriart2004fundamentals}
\textsc{Hiriart-Urruty, J.-B. and C.~Lemar{\'e}chal} (2004): \emph{Fundamentals
  of convex analysis}, Springer Science \& Business Media.

\bibitem[\protect\citeauthoryear{Kamenica and Gentzkow}{Kamenica and
  Gentzkow}{2011}]{kamenica2011bayesian}
\textsc{Kamenica, E. and M.~Gentzkow} (2011): \enquote{Bayesian persuasion,}
  \emph{American Economic Review}, 101, 2590--2615.

\bibitem[\protect\citeauthoryear{Kneeland}{Kneeland}{2015}]{kneeland2015identifying}
\textsc{Kneeland, T.} (2015): \enquote{Identifying higher-order rationality,}
  \emph{Econometrica}, 83, 2065--2079.

\bibitem[\protect\citeauthoryear{Koh}{Koh}{2022}]{koh2022stable}
\textsc{Koh, P.~S.} (2022): \enquote{Stable Outcomes and Information in Games:
  An Empirical Framework,} \emph{arXiv preprint arXiv:2205.04990}.

\bibitem[\protect\citeauthoryear{Lin and Liu}{Lin and
  Liu}{2022}]{lin2022credible}
\textsc{Lin, X. and C.~Liu} (2022): \enquote{Credible persuasion,} \emph{arXiv
  preprint arXiv:2205.03495}.

\bibitem[\protect\citeauthoryear{Magnolfi and Roncoroni}{Magnolfi and
  Roncoroni}{2019}]{magnolfi2019estimation}
\textsc{Magnolfi, L. and C.~Roncoroni} (2019): \enquote{Estimation of discrete
  games with weak assumptions on information,} \emph{University of Warwick,
  Department of Economics}.

\bibitem[\protect\citeauthoryear{Matthews}{Matthews}{1984}]{matthews1984implementability}
\textsc{Matthews, S.~A.} (1984): \enquote{On the implementability of reduced
  form auctions,} \emph{Econometrica: Journal of the Econometric Society},
  1519--1522.

\bibitem[\protect\citeauthoryear{Morris}{Morris}{2020}]{morris2020no}
\textsc{Morris, S.~E.} (2020): \enquote{No trade and feasible joint posterior
  beliefs,} Tech. rep., Massachusetts Institute of Technology.

\bibitem[\protect\citeauthoryear{Myerson}{Myerson}{1982}]{myerson1982optimal}
\textsc{Myerson, R.~B.} (1982): \enquote{Optimal coordination mechanisms in
  generalized principal--agent problems,} \emph{Journal of Mathematical
  Economics}, 10, 67--81.

\bibitem[\protect\citeauthoryear{Rehbeck}{Rehbeck}{2023}]{rehbeck2023revealed}
\textsc{Rehbeck, J.} (2023): \enquote{Revealed Bayesian expected utility with
  limited data,} \emph{Journal of Economic Behavior \& Organization}, 207,
  81--95.

\bibitem[\protect\citeauthoryear{Strassen}{Strassen}{1965}]{strassen1965existence}
\textsc{Strassen, V.} (1965): \enquote{The existence of probability measures
  with given marginals,} \emph{The Annals of Mathematical Statistics}, 36,
  423--439.

\bibitem[\protect\citeauthoryear{Syrgkanis, Tamer, and Ziani}{Syrgkanis
  et~al.}{2017}]{syrgkanis2017inference}
\textsc{Syrgkanis, V., E.~Tamer, and J.~Ziani} (2017): \enquote{Inference on
  auctions with weak assumptions on information,} \emph{arXiv preprint
  arXiv:1710.03830}.

\bibitem[\protect\citeauthoryear{Vohra, Toikka, and Vohra}{Vohra
  et~al.}{2023}]{toikka2022bayesian}
\textsc{Vohra, A., J.~Toikka, and R.~Vohra} (2023): \enquote{Bayesian
  persuasion: Reduced form approach,} \emph{Journal of Mathematical Economics},
  102863.

\end{thebibliography}
\appendix

\section{Omitted proofs}\label{appendix:omitted}
\subsection{Proof of \autoref{theorem:bce-c}}\label{appendix:bce-c}
\paragraph{Preliminaries} Before stating the proof of \autoref{theorem:bce-c}, we collect some definitions and results from convex analysis that we use in the proof.

Define $\Directions(\aactionb)=\{(-e_\type)_{\type\in\Types},(-\deviationbb)_{\aactionbb\in\Actions}\}$ to be the normal directions to the polyhedron \optab, which is implicitly defined as the set of vectors \vectors\ in $\reals^{\cardstates}$ that satisfy:
\begin{align}\label{eq:H-optab}
(\forall\type\in\Types)(-e_\type)^T\vectors&\leq0\\
(\forall\aactionbb\in\Actions)(-\deviationbb)^T\vectors&\leq0.   \nonumber
\end{align}
We are omitting the condition that $\sum_{\type\in\Types}\vectors(\type)=1$, but this is irrelevant in what follows.

Recall that for $\vectors\in\reals^{\cardstates}$, the normal cone of \opta\ at \vectors, $\normal(\vectors|\opta)$, is defined as
\begin{align}
    \normal(\vectors|\opta)=\{\direction\in\reals^{\cardstates}:(\forall\vectorb\in\opta)\direction^T\vectorb\leq\direction^T\vectors\}.
\end{align}
That is, the normal cone of \opta\ at \vectors\ is the set of directions \direction\ for which \vectors\ solves $\max\{\direction^T\vectorb:\vectorb\in\opta\}$. Importantly, the normal cone of a polyhedron, like \opta, satisfies the following property. To state it, recall that given a set of points $\Directions$, the cone of $\Directions$ is defined as $\mathrm{cone}(\Directions)=\{\sum_{j=1}^J\alpha_j\direction_j:J<\infty,\direction_j\in C,\alpha_j\geq0\}$. 

\begin{lemma}[\protect{\citet[Example~5.2.6(b)]{hiriart2004fundamentals}}]\label{lemma:normal-cone}
    Suppose $\belief\in\opta$ and let $\binding(\belief)=\{\direction\in\Directions(\aaction):\direction^T\belief=0\}$. Then, $\normal(\belief|\opta)=\mathrm{cone}(\binding(\belief))$.
\end{lemma}

\begin{proof}[Proof of \autoref{theorem:bce-c}]
Necessity of Equations \ref{eq:bce-c-states} and \ref{eq:bce-c-actions} follows from \citet[Theorem 3]{strassen1965existence}.

We now argue sufficiency.  Given \autoref{theorem:strassen}, it suffices to show that Equations \ref{eq:bce-c-states} and \ref{eq:bce-c-actions} imply \autoref{eq:strassen} holds for all $\direction\in\reals^{\cardstates}$.

    For fixed \direction, we can write \autoref{eq:strassen} as follows:
    \begin{align}\label{eq:strassen-rw}
        \sum_{\aaction\in\Actions}\marginal(\aaction)\max_{\belief\in\opta}c^T(\belief-\prior)\geq0.
    \end{align}
    Thus, \autoref{eq:strassen} holds for all directions $\direction\in\reals^{\cardstates}$ if and only if
    \begin{align}\label{eq:min-max}
\min_{\direction\in\reals^{\cardstates}}\sum_{\aaction\in\Actions}\marginal(\aaction)\max_{\belief\in\opta}c^T(\belief-\prior)\geq0.
    \end{align}
    Note that we can replace $\opta$ for the set of extreme points of \opta, $\opt_E(\aaction)$ in \autoref{eq:min-max}. That is, 
\begin{align}\label{eq:min-max-e}
\min_{\direction\in\reals^{\cardstates}}\sum_{\aaction\in\Actions}\marginal(\aaction)\max_{\belief\in\opt_E(\aaction)}c^T(\belief-\prior)\geq0.
    \end{align}
Now, let $E=\prod\{\opt_E(\aaction):\aaction\in\Actions\}$. For $\bar{\belief}_e\equiv(\belief_{e,\aaction})_{\aaction\in\Actions}\in E$, let \[\Directions(\bar{\belief}_e)=\{\direction\in\reals^{\cardstates}:(\forall\aaction\in\Actions)\;\direction^T\belief_{e,\aaction}=\max_{\belief\in\opta}\direction^T\belief\}.\]
Then, we can write the left hand side of \autoref{eq:min-max-e} as follows:
\begin{align}\label{eq:min-min-e}
    \min_{\bar{\belief}_e\in E}\min_{\direction\in\Directions(\bar{\belief}_e)}\sum_{\aaction\in\Actions}\marginal(\aaction)\max_{\belief\in\opt_E(\aaction)}\direction^T(\belief-\prior).
\end{align}
Note that for each $\bar{\belief}_e\in E$
\begin{align}
    \Directions(\bar{\belief}_e)=\cap_{\aaction\in\Actions}\normal\left(\belief_{e,\aaction}|\opta\right),
\end{align}
and by \autoref{lemma:normal-cone}, $\normal\left(\belief_{e,\aaction}|\opta\right)\subseteq\Directions(\aaction)$. Thus, Equations \ref{eq:bce-c-states} and \ref{eq:bce-c-actions} ensure that the term inside $\min_{\bar{\belief}_e\in E}$ is non-negative, so that \autoref{eq:strassen} holds for all $\direction\in\reals^{\cardstates}$.
    
\end{proof}

\subsection{Proof of \autoref{proposition:ring-network}}\label{appendix:ring-network}
\begin{proof}[Proof of \autoref{proposition:ring-network}]
In the ring-network base game, for a joint distribution $\joint\in\Delta(\Actions\times\Types)$, the obedience constraints can be written as follows:
\begin{align*}
    (\forall\aaction_1,\aactionb_1\in\Actions_1)\sum_{\type\in\Types}\joint_{\Types\times\Actions_1}(\aaction_1,\type)\left(\rnpayoff(\aaction_1,\type)-\rnpayoff(\aactionb_1,\type)\right)&\geq0\\
    (\forall \playerindex\in\{2,\dots,\nplayers\})(\forall\actioni,\actionbi\in\Actionsi)\sum_{\aaction_{\playerindex-1}\in\Actions_{\playerindex-1}}\joint_{\Actions_{\playerindex-1},\playerindex}(\aaction,\type)\left(\rnpayoff(\aaction_{\playerindex-1},\actioni)-\rnpayoff(\aaction_{\playerindex-1},\actionbi)\right)&\geq0,
\end{align*}
where $\joint_{\Types\times\Actions_1}$ is the marginal of \joint\ over $\Types\times\Actions_1$ and similarly for $\playerindex\geq2$, $\joint_{\Actions_{\playerindex-1}\times\Actionsi}$ is the marginal of \joint\ over $\Actions_{\playerindex-1}\times\Actionsi$. Thus, it is immediate that the conditions in \autoref{proposition:ring-network} are necessary for $(\prior,\marginalp)$ to be M-BCE-consistent. 

For sufficiency, note that \autoref{theorem:bce-c} implies that under the conditions of \autoref{proposition:ring-network}, $\left(\joint_{\Types\times\Actions_1},\dots,\joint_{\Actions_{\nplayers-1}\times\Actions_{\nplayers}}\right)$ exist each of which satisfy the respective marginal conditions and obedience constraints. 

Given these distributions, define $\hat{\joint}\in\Delta(\Actions\times\Types)$ as follows: for each $(\aaction,\type)\in\Actions\times\Types$
\begin{align}
    \hat{\joint}(\aaction,\type)=\joint_{\Actions_1\times\Types}(\aaction_1,\type)\joint_{\Actions_1\times\Actions_2}(\aaction_2|\aaction_1)\times\dots\joint_{\Actions_{N-1}\times\Actions_N}(\aaction_N|\aaction_{N-1}),
\end{align}
where abusing notation we let for $\playerindex\geq2$, $\joint_{\Actions_{\playerindex-1}\times\Actionsi}(\cdot|\aaction_{\playerindex-1})$ denote the distribution $\joint_{\Actions_{\playerindex-1}\times\Actionsi}$ conditional on $\aactionb_{\playerindex-1}=\aaction_{\playerindex-1}$.

Note that $\hat{\joint}(\aaction,\type)$ satisfies the obedience constraints of player $1$ if and only if $\joint_{\Actions_1\times\Types}(\cdot)$ does. Indeed, for all $\aaction_1,\aactionb_1$, we have
\begin{align}    &\sum_{\aaction_{-1},\type}\hat{\joint}(\aaction_1,\aaction_{-1},\type)\left(\rnpayoff_1(\aaction_1,\type)-\rnpayoff_1(\aactionb_1,\type)\right)=\nonumber\\
    &\sum_{\type}\joint_{\Actions_1\times\Types}(\aaction_1,\type)\left(\rnpayoff_1(\aaction_1,\type)-\rnpayoff_1(\aactionb_1,\type)\right)\sum_{(\aaction_2,\dots,\aaction_\nplayers)}\prod_{\playerindex=2}^\nplayers\joint_{\Actions_{\playerindex-1}\times\Actionsi}(\actioni|\aaction_{\playerindex-1})=\nonumber\\
    &\sum_{\type}\joint_{\Actions_1\times\Types}(\aaction_1,\type)\left(\rnpayoff_1(\aaction_1,\type)-\rnpayoff_1(\aactionb_1,\type)\right).
\end{align}
Consider now player $\playerindex\geq 2$. For simplicity, fix $\playerindex=2$--the rest of the players follow immediately. Then, let $\aaction_2,\aactionb_2\in\Actions_2$. We want to check that \joint\ satisfies the obedience constraint of player $2$ if and only if $\joint_{\Actions_1\times\Actions_2}$ does.
\begin{align*}
    &\sum_{\aaction_{-2},\type}\hat{\joint}(\aaction_2,\aaction_{-2},\type)\left(\rnpayoff_2(\aaction_1,\aaction_2)-\rnpayoff_2(\aaction_1,\aactionb_2)\right)=\\
    &\sum_{\aaction_1,\type}\joint_{\Actions_1\times\Types}(\aaction_1,\type)\joint_{\Actions_1\times\Actions_2}(\aaction_2|\aaction_1)\left(\rnpayoff_2(\aaction_1,\aaction_2)-\rnpayoff_2(\aaction_1,\aactionb_2)\right)\sum_{(\aaction_3,\dots,\aaction_\nplayers)}\prod_{\playerindex=3}^\nplayers\joint_{\Actions_{\playerindex-1}\times\Actionsi}(\actioni|\aaction_{\playerindex-1})=\\
&\sum_{\aaction_1\in\Actions_1}\left(\sum_{\type}\joint_{\Actions_1\times\Types}(\aaction_1,\type)\right)\joint_{\Actions_1\times\Actions_2}(\aaction_2|\aaction_1)\left(\rnpayoff_2(\aaction_1,\aaction_2)-\rnpayoff_2(\aaction_1,\aactionb_2)\right)=\\
    &\sum_{\aaction_1\in\Actions_1}\marginalbase_{01}(\aaction_1)\joint_{\Actions_1\times\Actions_2}(\aaction_2|\aaction_1)\left(\rnpayoff_2(\aaction_1,\aaction_2)-\rnpayoff_2(\aaction_1,\aactionb_2)\right)=\\
    &\sum_{\aaction_1\in\Actions_1}\joint_{\Actions_1\times\Actions_2}(\aaction_1,\aaction_2)\left(\rnpayoff_2(\aaction_1,\aaction_2)-\rnpayoff_2(\aaction_1,\aactionb_2)\right),
\end{align*}
    where the third equality follows from the assumption that $\joint_{\Actions_1\times\Types}$ satisfies the marginal constraints for player $1$.
\end{proof}

\section{A demand-supply interpretation}\label{appendix:demand-supply}
We provide here an alternative, but still network based, proof of \autoref{theorem: single agent BCE} using the fundamental results of \cite{gale1957theorem} on demand and supply in a network. 

\paragraph{Flows in networks:} The problem in \cite{gale1957theorem} can be described as follows. Given a
graph $\left( V,E\right) $, suppose that to each node $v\in V$ corresponds a
real number $\demandbase(v)$. If $\demandbase(v)>0$ we interpret $\left\vert \demandbase(v)\right\vert $
as the \emph{demand} of node $v$ for some homogenous good. If $\demandbase(v)<0$ we
interpret $\left\vert \demandbase(v)\right\vert $ as the \emph{supply} of the good by $
v$. To each edge $(v,v^{\prime })\in E$ correspond two nonnegative real
numbers $c(v,v^{\prime })$ and $c\left( v^{\prime },v\right) $, the capacity
of this edge, which assign an upper bound to the possible flow of the good
from $v$ to $v^{\prime }$ and from $v^{\prime }$ to $v$, respectively. The
demand $\demand$ is called \emph{feasible} if there is a flow
in the graph such that the flow along each edge is no greater than its
capacity, and the net flow into (out of) each node is at least (at most)
equal to the demand (supply) at that node. The demand problem identifies the
conditions under which a given demand $\demand$ is feasible
in the graph.

\paragraph{BCE-consistency and the demand problem:} Fix a Bayes plausible distribution $\bsplit \in \Delta \left( \Delta \left( \type \right)
\right)$ and denote its support by $T=\mathrm{supp}\;\bsplit$. Because \Actions\ is finite, it is without loss of generality to assume that $T$ is a finite set \citep{myerson1982optimal,kamenica2011bayesian}. 

The conditions in \cite{gale1957theorem} can be used to check whether $\bsplit $ implements the action distribution $\marginal$: That is, that a decision rule \strat\ exists that together with \bsplit\ define an obedient experiment (see \autoref{eq:tau-feasibility}). To that end, we construct a (bipartite) graph in which posterior beliefs serve as supply nodes, and actions serve as demand nodes. That is, the homogeneous good in our construction can be thought of as probability that \textquotedblleft flows\textquotedblright\ from induced posterior beliefs to actions. The construction of the graph guarantees that if the demand is feasible, we can specify choices for the agent such that the probabilities according to which she breaks ties between optimal actions at each posterior belief induce the (ex-ante) desired action distribution $\marginal$.

Formally, define the graph $\graphp(\bsplit)=\left( \Actions\cup T,E\right)$ as follows.  To each action $\aaction\in\Actions$ corresponds a node that demands the marginal probability of \aaction, i.e. $\demand\left( \aaction\right) =\marginal\left( \aaction\right)$. To each belief $\belief \in T$ corresponds a node that supplies the probability with which $\belief $ is realized in $\bsplit $, i.e. $\demand\left( \belief \right) =-\bsplit \left(\belief \right)$. For any belief-action pair $(\belief,\aaction)$, an edge $\left( \belief ,\aaction\right) \in E$ exists between the nodes $\belief $ and $a$ if and only if  action $a$ is optimal under posterior belief $\belief $, that is if and only if, $\aaction\in\aaction^{\ast }\left(\belief \right)$. Finally, for any edge $\left( \belief ,\aaction\right) \in E$, the edge's flow capacity is given by $c\left( \belief ,\aaction\right) =\infty $ and $c\left( a,\belief \right) =0$. That is, there is no upper bound on the flow from $\belief $ to $a$, but there cannot be a flow from $a$ to $\belief $. We denote the flow from node $\belief $ to node $a$ by $f\left( \belief ,\aaction\right) $. If there is no edge between $\belief $ and $a$ then $f\left( \belief ,\aaction\right) =0$. The right-hand side panel of Figure \ref{fig:graph_example} illustrates the graph $\graphp$.

\begin{figure}[h!]
    \begin{subfigure}{0.45\textwidth}
        \begin{center}
            \resizebox{5.5cm}{!}{\tikzset{every picture/.style={line width=0.75pt}} 

\begin{tikzpicture}[x=0.75pt,y=0.75pt,yscale=-1,xscale=1]

\draw  [fill={rgb, 255:red, 74; green, 144; blue, 226 }  ,fill opacity=0.2 ] (230,250) -- (230,170) -- (180,132) -- (100,250) -- cycle ;
\draw   (235,50) -- (370,250) -- (100,250) -- cycle ;
\draw  [fill={rgb, 255:red, 208; green, 2; blue, 27 }  ,fill opacity=0.2 ] (286,126) -- (370,250) -- (230,250) -- (230,170) -- cycle ;
\draw  [fill={rgb, 255:red, 0; green, 0; blue, 0 }  ,fill opacity=1 ] (257,158) .. controls (257,155.24) and (259.24,153) .. (262,153) .. controls (264.76,153) and (267,155.24) .. (267,158) .. controls (267,160.76) and (264.76,163) .. (262,163) .. controls (259.24,163) and (257,160.76) .. (257,158) -- cycle ;
\draw  [fill={rgb, 255:red, 74; green, 144; blue, 226 }  ,fill opacity=1 ] (281,126) .. controls (281,123.24) and (283.24,121) .. (286,121) .. controls (288.76,121) and (291,123.24) .. (291,126) .. controls (291,128.76) and (288.76,131) .. (286,131) .. controls (283.24,131) and (281,128.76) .. (281,126) -- cycle ;
\draw  [fill={rgb, 255:red, 74; green, 144; blue, 226 }  ,fill opacity=1 ] (175,132) .. controls (175,129.24) and (177.24,127) .. (180,127) .. controls (182.76,127) and (185,129.24) .. (185,132) .. controls (185,134.76) and (182.76,137) .. (180,137) .. controls (177.24,137) and (175,134.76) .. (175,132) -- cycle ;
\draw  [fill={rgb, 255:red, 74; green, 144; blue, 226 }  ,fill opacity=1 ] (230,50) .. controls (230,47.24) and (232.24,45) .. (235,45) .. controls (237.76,45) and (240,47.24) .. (240,50) .. controls (240,52.76) and (237.76,55) .. (235,55) .. controls (232.24,55) and (230,52.76) .. (230,50) -- cycle ;
\draw  [fill={rgb, 255:red, 74; green, 144; blue, 226 }  ,fill opacity=1 ] (365,250) .. controls (365,247.24) and (367.24,245) .. (370,245) .. controls (372.76,245) and (375,247.24) .. (375,250) .. controls (375,252.76) and (372.76,255) .. (370,255) .. controls (367.24,255) and (365,252.76) .. (365,250) -- cycle ;
\draw  [fill={rgb, 255:red, 74; green, 144; blue, 226 }  ,fill opacity=1 ] (225,250) .. controls (225,247.24) and (227.24,245) .. (230,245) .. controls (232.76,245) and (235,247.24) .. (235,250) .. controls (235,252.76) and (232.76,255) .. (230,255) .. controls (227.24,255) and (225,252.76) .. (225,250) -- cycle ;
\draw  [fill={rgb, 255:red, 74; green, 144; blue, 226 }  ,fill opacity=1 ] (96,249) .. controls (96,246.24) and (98.24,244) .. (101,244) .. controls (103.76,244) and (106,246.24) .. (106,249) .. controls (106,251.76) and (103.76,254) .. (101,254) .. controls (98.24,254) and (96,251.76) .. (96,249) -- cycle ;
\draw  [fill={rgb, 255:red, 74; green, 144; blue, 226 }  ,fill opacity=1 ] (225,170) .. controls (225,167.24) and (227.24,165) .. (230,165) .. controls (232.76,165) and (235,167.24) .. (235,170) .. controls (235,172.76) and (232.76,175) .. (230,175) .. controls (227.24,175) and (225,172.76) .. (225,170) -- cycle ;

\draw (264,156.4) node [anchor=north west][inner sep=0.75pt]    {$\prior$};
\draw (225,22.4) node [anchor=north west][inner sep=0.75pt]    {$\belief _{1}$};
\draw (372,253.4) node [anchor=north west][inner sep=0.75pt]    {$\belief _{2}$};
\draw (79,250.4) node [anchor=north west][inner sep=0.75pt]    {$\belief _{3}$};
\draw (157,204.4) node [anchor=north west][inner sep=0.75pt]    {$\opt( \aaction_{3})$};
\draw (269,204.4) node [anchor=north west][inner sep=0.75pt]    {$\opt( \aaction_{2})$};
\draw (213,100.4) node [anchor=north west][inner sep=0.75pt]    {$\opt( \aaction_{1})$};
\draw (293,105.4) node [anchor=north west][inner sep=0.75pt]    {$\belief _{12}$};
\draw (217,140.4) node [anchor=north west][inner sep=0.75pt]    {$\belief _{123}$};
\draw (220,260.4) node [anchor=north west][inner sep=0.75pt]    {$\belief _{23}$};
\draw (153,108.4) node [anchor=north west][inner sep=0.75pt]    {$\belief _{13}$};

\end{tikzpicture}}
        \end{center}
    \caption{}
    \end{subfigure}\quad\quad
    \begin{subfigure}{0.45\textwidth}
        \begin{center}
            \resizebox{7cm}{!}{\tikzset{every picture/.style={line width=0.75pt}} 

\begin{tikzpicture}[x=0.75pt,y=0.75pt,yscale=-1,xscale=1]

\draw    (119.6,50.6) -- (188,50.6) ;
\draw    (119.6,50.6) -- (188,100.6) ;
\draw    (119.6,100.6) -- (188,100.6) ;
\draw    (119.6,150.6) -- (188,150.6) ;
\draw    (119.6,150.6) -- (188,100.6) ;
\draw    (119.6,150.6) -- (188,50.6) ;
\draw  [fill={rgb, 255:red, 74; green, 144; blue, 226 }  ,fill opacity=1 ] (114.6,50.6) .. controls (114.6,47.84) and (116.84,45.6) .. (119.6,45.6) .. controls (122.36,45.6) and (124.6,47.84) .. (124.6,50.6) .. controls (124.6,53.36) and (122.36,55.6) .. (119.6,55.6) .. controls (116.84,55.6) and (114.6,53.36) .. (114.6,50.6) -- cycle ;
\draw  [fill={rgb, 255:red, 74; green, 144; blue, 226 }  ,fill opacity=1 ] (114.6,100.6) .. controls (114.6,97.84) and (116.84,95.6) .. (119.6,95.6) .. controls (122.36,95.6) and (124.6,97.84) .. (124.6,100.6) .. controls (124.6,103.36) and (122.36,105.6) .. (119.6,105.6) .. controls (116.84,105.6) and (114.6,103.36) .. (114.6,100.6) -- cycle ;
\draw  [fill={rgb, 255:red, 74; green, 144; blue, 226 }  ,fill opacity=1 ] (114.6,150.6) .. controls (114.6,147.84) and (116.84,145.6) .. (119.6,145.6) .. controls (122.36,145.6) and (124.6,147.84) .. (124.6,150.6) .. controls (124.6,153.36) and (122.36,155.6) .. (119.6,155.6) .. controls (116.84,155.6) and (114.6,153.36) .. (114.6,150.6) -- cycle ;
\draw  [fill={rgb, 255:red, 208; green, 2; blue, 27 }  ,fill opacity=1 ] (183,50.6) .. controls (183,47.84) and (185.24,45.6) .. (188,45.6) .. controls (190.76,45.6) and (193,47.84) .. (193,50.6) .. controls (193,53.36) and (190.76,55.6) .. (188,55.6) .. controls (185.24,55.6) and (183,53.36) .. (183,50.6) -- cycle ;
\draw  [fill={rgb, 255:red, 208; green, 2; blue, 27 }  ,fill opacity=1 ] (183,100.6) .. controls (183,97.84) and (185.24,95.6) .. (188,95.6) .. controls (190.76,95.6) and (193,97.84) .. (193,100.6) .. controls (193,103.36) and (190.76,105.6) .. (188,105.6) .. controls (185.24,105.6) and (183,103.36) .. (183,100.6) -- cycle ;
\draw  [fill={rgb, 255:red, 208; green, 2; blue, 27 }  ,fill opacity=1 ] (183,150.6) .. controls (183,147.84) and (185.24,145.6) .. (188,145.6) .. controls (190.76,145.6) and (193,147.84) .. (193,150.6) .. controls (193,153.36) and (190.76,155.6) .. (188,155.6) .. controls (185.24,155.6) and (183,153.36) .. (183,150.6) -- cycle ;

\draw (93.2,87.4) node [anchor=north west][inner sep=0.75pt]    {$\mu _{2}$};
\draw (87.2,37.2) node [anchor=north west][inner sep=0.75pt]    {$\mu _{12}$};
\draw (81.2,137.6) node [anchor=north west][inner sep=0.75pt]    {$\mu _{123}$};
\draw (202.2,39.6) node [anchor=north west][inner sep=0.75pt]    {$a_{1}$};
\draw (202.2,90.2) node [anchor=north west][inner sep=0.75pt]    {$a_{2}$};
\draw (202.2,140.8) node [anchor=north west][inner sep=0.75pt]    {$a_{3}$};
\draw (232.8,40.4) node [anchor=north west][inner sep=0.75pt]  [color={rgb, 255:red, 208; green, 2; blue, 27 }  ,opacity=1 ]  {$\nu _{0}( a_{1})$};
\draw (232,90.5) node [anchor=north west][inner sep=0.75pt]  [color={rgb, 255:red, 208; green, 2; blue, 27 }  ,opacity=1 ]  {$\nu _{0}( a_{2})$};
\draw (233.2,140.6) node [anchor=north west][inner sep=0.75pt]  [color={rgb, 255:red, 208; green, 2; blue, 27 }  ,opacity=1 ]  {$\nu _{0}( a_{3})$};
\draw (16,36.8) node [anchor=north west][inner sep=0.75pt]  [color={rgb, 255:red, 74; green, 144; blue, 226 }  ,opacity=1 ]  {$-\tau ( \mu _{12})$};
\draw (16,87) node [anchor=north west][inner sep=0.75pt]  [color={rgb, 255:red, 74; green, 144; blue, 226 }  ,opacity=1 ]  {$-\tau ( \mu _{2})$};
\draw (16,137.2) node [anchor=north west][inner sep=0.75pt]  [color={rgb, 255:red, 74; green, 144; blue, 226 }  ,opacity=1 ]  {$-\tau ( \mu _{123})$};
\draw (30.5,4.8) node [anchor=north west][inner sep=0.75pt]  [color={rgb, 255:red, 74; green, 144; blue, 226 }  ,opacity=1 ]  {$supply$};
\draw (223.3,4.8) node [anchor=north west][inner sep=0.75pt]  [color={rgb, 255:red, 208; green, 2; blue, 27 }  ,opacity=1 ]  {$demand$};

\end{tikzpicture}}
        \end{center}
    \caption{}
    \end{subfigure}
    \caption{ Illustration of the supply-demand proof of \autoref{theorem: single agent BCE} with $|A|=\cardstates=3$.
    The simplex on the left-hand side depicts the optimal action(s) for each posterior belief. The graph on the right-hand side corresponds to the Bayes plausible distribution over posteriors $\bsplit$ supported on $T=\{\belief_{12}, \belief_{2}, \belief_{123}\}$.}
    \label{fig:graph_example}
\end{figure}
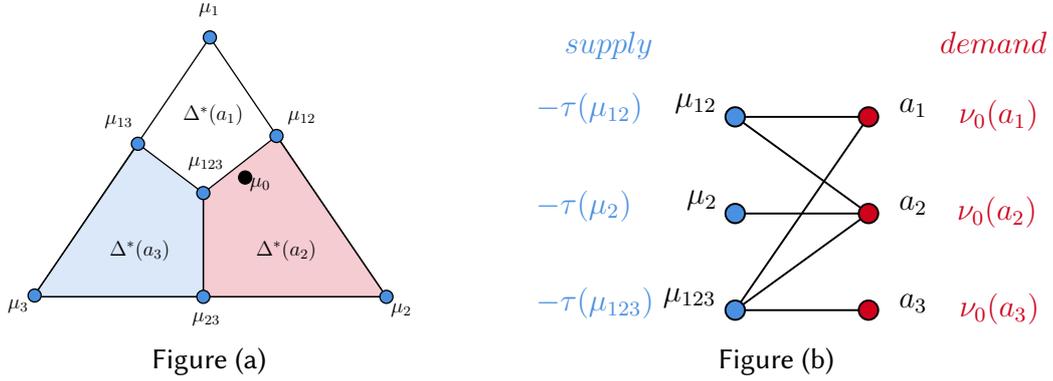

\autoref{proposition:equivalence of representation} motivates the connection between our problem and that in \cite{gale1957theorem}. 
\begin{proposition}[Feasibility and BCE-consistency]\label{proposition:equivalence of representation}
The Bayes plausible distribution over posteriors \bsplit\ implements \marginal\ if and only if \demand\ is feasible on $\graphp(\bsplit)$.   
\end{proposition}

The proof of \autoref{proposition:equivalence of representation} relies on the following lemma:
\begin{lemma}[Market clearing]\label{lemma:exact satiation}
If $\demand$ is feasible on $\graphp(\bsplit)$, then
the flow out of any supply node $\belief \in T$ is exactly $\bsplit \left( \belief
\right) $ (and not less), and the flow into any demand node $\aaction\in\Actions$ is
exactly $\marginal\left( \aaction\right) $ (and not more).
\end{lemma}

\begin{proof}[Proof of \autoref{lemma:exact satiation}]
Suppose that $\demand$ is feasible. We  show that the
flow into any demand node $\aaction\in\Actions$ is exactly $\marginal\left( \aaction\right) $.
Towards a contradiction, suppose that $\sum_{\belief \in T}f\left( \belief
,\aaction\right) \geq \marginal\left( \aaction\right) $ for all $\aaction\in\Actions$, with  strict
inequality for some $a$. Summing over all actions on both sides of the
inequality yields \[\sum_{\aaction\in\Actions}\sum_{\belief \in T}f\left( \belief ,\aaction\right)
>\sum_{\aaction\in\Actions}\marginal\left( \aaction\right) =1.\] On the other hand, because $
\demand$ is feasible, then the flow out of each $
\belief \in T$ is at most $\bsplit \left( \belief \right) $, and therefore for all $\belief \in T$ \[\sum_{\aaction\in
A}f\left( \belief ,\aaction\right) \leq \bsplit \left( \belief \right) .\]
Summing again over all actions on both sides yields \[\sum_{\belief \in
T}\sum_{\aaction\in\Actions}f\left( \belief ,\aaction\right) \leq \sum_{\belief \in T}\bsplit \left( \belief
\right) =1,\]  a contradiction. The proof that the flow out of any
supply node $\belief $ is exactly $\bsplit \left( \belief \right) $ is analogous and
hence omitted.
\end{proof}

\begin{proof}[Proof of \autoref{proposition:equivalence of representation}]

Suppose first that the Bayes plausible distribution over posteriors \bsplit\ is such that \demand\ is feasible on $\graphp(\bsplit)$ and let $f$ denote the feasible flow. Consider a decision rule $\strat:\Posteriors\mapsto\Delta(\Actions)$ such that the agent takes action $\aaction\in\Actions$ when the belief is $\belief \in T$ with probability $\sigma\left(
a~\ |~\belief \right) =f\left( \belief ,\aaction\right) /\bsplit \left( \belief \right) $. This correctly defines a decision rule as
\[
\sum_{\aaction\in\Actions}\strat \left( a~\ |~\belief \right) =\frac{\sum_{\aaction\in\Actions}f\left( \belief
,\aaction\right) }{\bsplit \left( \belief \right) }=1
\]
where the second equality is implied by \autoref{lemma:exact satiation}. 
Furthermore, \strat\ is optimal for the agent because $\belief $ and $a$ are connected with
an edge only if $a$ is optimal under $\belief $, i.e. $\aaction\in\aaction^{\ast }\left( \belief
\right) $.

To verify that $(\bsplit,\strat)$ induce \marginal, note that for 
all $\aaction\in\Actions$
\[
\sum_{\belief \in T}\bsplit \left( \belief \right) \strat \left( a~\
|~\belief \right) =\sum_{\belief \in T}f\left( \belief ,\aaction\right) =\marginal\left(
\aaction\right) \text{.}
\]%
where the second equality follows again from \autoref{lemma:exact satiation}. Thus, $\marginal$ is consistent with $\bsplit$.

Conversely, suppose that $\pair$ are BCE-consistent. Then, by \autoref{eq:tau-feasibility}, a Bayes plausible distribution over posteriors \bsplit\ and a decision rule \strat\ exists that induce an obedient experiment.\footnote{Namely, BCE-consistency implies the existence of an obedient experiment from which we can infer the following distribution over posteriors. First, let
\begin{align*}
    \belief_a(\type)=\frac{\prior(\type)\joint(\aaction|\type)}{\sum_{\typeb\in\Types}\prior(\typeb)\joint(\aaction|\typeb)},
\end{align*}
and let $\bsplit(\{\belief_a\})=\sum_{\type\in\Types}\prior(\type)\joint(\aaction|\type)$. The decision rule $\strat(\cdot|\belief_a)=\mathbbm{1}[\aactionb=\aaction]$ completes the construction.} Define the graph $\graphp(\bsplit)$ and the demand \demand. Note that the demand \demand\ is feasible on $\graphp(\bsplit)$ by defining the flow $f\left( \belief ,\aaction\right) =\strat \left( a~\ |~\belief \right) \bsplit \left( \belief \right) 
$ for all $\left( \belief ,\aaction\right) \in T\times A$.
\end{proof}

\autoref{proposition:equivalence of representation} implies that verifying that \bsplit\ implements \marginal\ is equivalent to verifying the feasibility of the demand $\demand$ for the graph $\graphp$. The main theorem in \cite{gale1957theorem} provides necessary and sufficient conditions under which $\demand$ is feasible. Adapted to our setting, the
conditions in \cite{gale1957theorem} can be stated as follows:
\begin{proposition}[\citealp{gale1957theorem}]
\label{proposition: Gale 1957}The demand $\demand$ is feasible on graph $\graphp(\bsplit)$ if
and only if for every set $\Actionsb\subseteq\Actions$ a flow $f_{\Actionsb}$ exists such that:
\begin{enumerate}
\item\label{itm:gale-1} $\sum_{\aaction\in\Actions}f_{\Actionsb}\left( \belief ,\aaction\right) \leq \bsplit \left( \belief \right) $
\quad for all $\belief \in T$, and
\item\label{itm:gale-2} $\sum_{\aaction\in\Actionsb}\sum_{\belief \in T}f_{\Actionsb}\left( \belief ,\aaction\right) \geq
\sum_{\aaction\in\Actionsb}\marginal\left( \aaction\right) $.
\end{enumerate}
\end{proposition}

In our setting, given a set $\Actionsb\subseteq\Actions$ items \ref{itm:gale-1} and \ref{itm:gale-2} in \autoref{proposition: Gale 1957} are satisfied for some flow $f_{\Actionsb}$ if and only if they are satisfied when the out flow from every supply node that is connected to nodes in $\Actionsb$ is maximal. Denote the set of all posterior beliefs in $T$ for which some action in $\Actionsb$ is optimal (and perhaps also actions that are not in $\Actionsb$) by $\opt \left( \Actionsb\right) =\left\{ \belief \in T \mid \exists \aaction\in\Actionsb,\aaction\in\aaction^{\ast }\left( \belief \right) \right\} $.  
Thus, in the graph we constructed, all and only beliefs (i.e., supply nodes) in $\opt  \left( \Actionsb\right) $ are connected to actions (i.e., demand nodes) in $\Actionsb$. The next corollary follows immediately:
\begin{corollary}\label{corollary:gale}
The Bayes plausible distribution over posteriors \bsplit\ implements \marginal\ if and only if for every subset $\Actionsb\subseteq\Actions$,
\begin{align}\label{eq:demand}
\sum_{\belief \in \opt  \left( \Actionsb\right) }\bsplit \left( \belief \right) \geq
\sum_{\aaction\in\Actionsb}\marginal\left( \aaction\right) \text{.}  
\end{align}
\end{corollary}
To see that the condition in \autoref{corollary:gale} is equivalent to that in \autoref{theorem: single agent BCE}, note first
that because $\bsplit,\marginal$ are measures (and hence add up to $1$), \autoref{eq:demand} can be
equivalently written as follows:%
\begin{align}\label{eq:theorem-1}
\sum_{\aaction\in \overline{\Actionsb}}\marginal\left( \aaction\right) \geq \sum_{\belief \in\overline{
\opt  \left( \Actionsb\right)} }\bsplit \left( \belief \right) \text{,}
\end{align}
where the upper-bar notation denotes the complement of a set--for instance, $\overline{\Actionsb}=\Actions\setminus\Actionsb$. 

Note that
\begin{align*}
    \overline{\opt(\Actionsb)}=\{\belief\in T|\aaction^*(\belief)\cap \Actionsb=\emptyset\}=    \bigcup_{C\subseteq\overline{\Actionsb}}\{\belief\in T|a^*(\belief)=C\}.
\end{align*}
Hence, we can write \autoref{eq:theorem-1} as follows
\begin{align}
\sum_{\aaction\in\overline{\Actionsb}}\marginal(\aaction)\geq\sum_{C\subseteq\overline{\Actionsb}}\sum_{\belief\in T:\aaction^*(\belief)=C}\bsplit(\belief)
\end{align}
which is the equation in \autoref{theorem: single agent BCE}.

\end{document}